\documentclass[runningheads,final]{llncs}

\usepackage[utf8]{inputenc} %
\usepackage[T1]{fontenc}    %
\usepackage{graphicx} %
\usepackage{mathtools}
\usepackage{hyperref}       %
\usepackage{url}            %
\usepackage{booktabs}       %
\usepackage{amsfonts}       %
\usepackage{nicefrac}       %
\usepackage{microtype}      %
\usepackage{xcolor}         %
\usepackage{amsmath}
\usepackage{listings}
\usepackage{stmaryrd} %
\usepackage{xspace}
\usepackage{proof}  %
\usepackage{stackengine} %
\usepackage{pifont}
\usepackage{subfig} %
\newcommand{\nmapsto}{\not\mapsto}
\usepackage{apxproof}
\renewcommand{\appendixprelim}[1]{%
}

\newcommand{\dstechnical}{\includegraphics[width=0.9em]{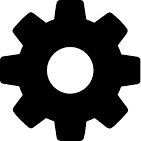}}

\usepackage[inline]{enumitem}
\usepackage{todonotes}

\newcommand{\SEM}[2]{\llbracket #1 \rrbracket #2}
\newcommand{\True}{\mathit{true}}
\newcommand{\False}{\mathit{false}}
\newcommand{\Val}{v}
\newcommand{\Vals}{\mathcal{V}}

\newcommand*{\Bools}{\texttt{bool}\xspace}
\newcommand*{\Ints}{\texttt{int}\xspace}
\newcommand*{\Reals}{\texttt{real}\xspace}
\newcommand*{\LTLs}{\texttt{ltl\_t}\xspace}

\newcommand*{\Z}{\mathbb{Z}}
\newcommand*{\R}{\mathbb{R}}
\DeclareMathOperator{\dom}{dom}

\newcommand*{\HOAone}{\textsc{HOA}\xspace}
\newcommand*{\HOApp}{\textsc{HOApp}\xspace}
\newcommand{\SET}{\bfseries\raisebox{.5pt}{\textcircled{\raisebox{-.9pt} {\footnotesize\textsf 0}}}}

\newtheoremrep{thm}{Theorem}

\title{Symbolic \texorpdfstring{$\omega$}{Omega}-Automata with Obligations}
\titlerunning{Symbolic $\omega$-Automata with Obligations}
\author{Luca {Di Stefano}\orcidID{0000-0003-1922-3151}}
\authorrunning{L. {Di Stefano}}

\institute{%
TU Wien, Institute of Computer Engineering, Treitlstraße 3, 1040 Vienna, Austria\\
\email{luca.di.stefano@tuwien.ac.at}%
}

\usepackage{tikz}
\usetikzlibrary {automata,positioning,fit,calc,patterns,shapes,backgrounds}

\makeatletter
\pgfkeys{/pgf/document dog ear/.value required}%
\pgfkeys{/pgf/document dog ear/.code={%
        \pgfmathparse{#1}%
        \edef\pgfmathresult{\pgfmathresult pt}%
        \pgfkeyslet{/pgf/document dog ear}{\pgfmathresult}%
    }%
}%
\pgfkeys{/pgf/document dog ear=7.5pt}%
\pgfdeclareshape{document}{
    \inheritsavedanchors[from=rectangle] %
    \inheritanchorborder[from=rectangle]
    \inheritanchor[from=rectangle]{center}
    \inheritanchor[from=rectangle]{north}
    \inheritanchor[from=rectangle]{south}
    \inheritanchor[from=rectangle]{west}
    \inheritanchor[from=rectangle]{east}
    \backgroundpath{%
        \southwest \pgf@xa=\pgf@x \pgf@ya=\pgf@y
        \northeast \pgf@xb=\pgf@x \pgf@yb=\pgf@y
        \pgf@xc=\pgf@xb \advance\pgf@xc by-\pgfkeysvalueof{/pgf/document dog ear}
        \pgf@yc=\pgf@yb \advance\pgf@yc by-\pgfkeysvalueof{/pgf/document dog ear}
        \pgfpathmoveto{\pgfpoint{\pgf@xa}{\pgf@ya}}
        \pgfpathlineto{\pgfpoint{\pgf@xa}{\pgf@yb}}
        \pgfpathlineto{\pgfpoint{\pgf@xc}{\pgf@yb}}
        \pgfpathlineto{\pgfpoint{\pgf@xb}{\pgf@yc}}
        \pgfpathlineto{\pgfpoint{\pgf@xb}{\pgf@ya}}
        \pgfpathclose
        \pgfpathmoveto{\pgfpoint{\pgf@xc}{\pgf@yb}}
        \pgfpathlineto{\pgfpoint{\pgf@xc}{\pgf@yc}}
        \pgfpathlineto{\pgfpoint{\pgf@xb}{\pgf@yc}}
        \pgfpathlineto{\pgfpoint{\pgf@xc}{\pgf@yc}}
    }%
}%
\makeatother

\begin{document}

\maketitle

\begin{abstract}
Extensions of $\omega$-automata to infinite alphabets typically
rely on symbolic guards to keep the transition relation finite,
and on registers or memory cells to preserve information from past symbols.
Symbolic transitions alone are ill-suited to act on this information,
and register automata have intricate formal semantics and issues with tractability.
We propose a slightly different approach based on \emph{obligations},
i.e., assignment-like constructs attached to transitions.
Whenever a transition with an obligation is taken, the obligation is evaluated on the \emph{current} symbol and yields a constraint on the \emph{next} symbol that the automaton will read.

We formalize obligation automata with existential and universal branching
and Emerson-Lei acceptance conditions, which subsume classic families such as B\"uchi, Rabin, Strett, and parity
automata.
We show that these automata recognise a strict superset of $\omega$-regular languages.

To illustrate the practicality of our proposal, we also introduce a machine-readable format to express
obligation automata and describe a tool implementing several operations over them,
including automata product and emptiness checking.

\end{abstract}

\section{Introduction}
Automata over infinite words (or $\omega$-automata) are a cornerstone of
theoretical computer science, as they enable representing and reasoning about
computing systems, protocols, linear temporal properties, and game structures;
thus, they are widely applied in domains ranging from formal verification,
to logical satisfiability checking, to reactive synthesis.
However, these automata only operate on finite alphabets and thus
are ill-suited to contexts that involve infinite-state domains,
which are becoming extremely relevant to the research community of today~\cite{DBLP:conf/cav/AzzopardiSPS25,DBLP:journals/pacmpl/HeimD24,DBLP:conf/fmcad/MaderbacherB22}.
\emph{Symbolic automata} represent a first step towards handling infinite
alphabets.
Transitions in a symbolic automaton are not guarded by mere symbols,
but rather by predicates from an alphabet theory.
This allows each single symbolic transition to handle a possibly
infinite set of symbols.
However, the automaton's finite state space leads to an inevitable
limitation:
after reading one of infinitely many symbols, the automaton must pick a
successor state from a finite set of choices.
By the pigeonhole principle, there will be
infinitely many symbols for which the automaton chooses the
same successor state, losing some information in the process.
Fully symbolic models, such as those expressible in nuXmv~\cite{DBLP:conf/cav/CavadaCDGMMMRT14}
or MoXi~\cite{DBLP:conf/cav/JohannsenNDISTVR24},
overcome this limitation by also treating the (now infinite) state space
symbolically.
This begs the question: can we seize at least some of the expressive power of
these symbolic models while retaining a finite state space,
and thus the rich analytical toolbox of $\omega$-automata theory?
Previous work, discussed below, answered through
\emph{register automata}~\cite{DBLP:journals/tcs/KaminskiF94} that are able to store and recall data from memory cells.
This framework yields complicated semantics and highly complex
decision problems on infinite words.

In this work, we suggest an alternative approach by introducing \emph{obligation automata}.
An obligation automaton is a symbolic automaton: it reads
valuations of variables from a (finite) set $V$,
and its transitions are guarded by predicates over said set.
In addition, a transition may be labelled by an \emph{obligation}.
If such a transition is taken, the obligation constrains the
\emph{next} valuations that the automaton can read, based on the
\emph{current} valuation. This construct is reminiscent of
how  languages such as
SMV~\cite{DBLP:conf/cav/ClarkeMCH96}
and BTOR~\cite{DBLP:conf/cav/NiemetzPWB18}
use \texttt{next($\cdot$)} to declare symbolic transition relations;
its simplicity lets us easily define formal semantics for Emerson-Lei
acceptance conditions with existential and universal branching.
As we will see, obligations are a direct generalisation
of $\omega$-automata (Theorem~\ref{thm:fragment}) and recognise a superset of $\omega$-regular languages (Theorem~\ref{thm:finite}).

On the practical side,
we also contribute a dialect of the
popular Hanoi Omega-Automata format (HOA)~\cite{DBLP:conf/cav/BabiakBDKKM0S15}
to express obligation automata over integer and real variables.
We call this format ``HOA plus plus'' (\HOApp).
Lastly, we introduce a prototype tool that can validate and perform several useful operations
on \HOApp automata. These include computing the product of two automata, translating an LTL formula
with predicates into an automaton, and checking the emptiness of an automaton via a reduction to model checking.

This work is structured as follows.
{
Section~\ref{sec:oblaut} introduces obligation automata and their semantics.
In Section~\ref{sec:hoapp} we describe our proposed dialect of HOA for
obligation automata, while
Section~\ref{sec:tool} describes the capabilities of our prototype tool to handle this format.
Lastly, Section~\ref{sec:conclusion} contains our concluding remarks.}

\noindent\emph{Related work.}
The common approach to model memory in automata or temporal logics is
to introduce the concept of a \emph{register} or \emph{memory cell}.
This approach is seen, for instance, in register automata (RA), also known as finite-memory automata~\cite{DBLP:journals/tcs/KaminskiF94},
and Temporal Stream Logic (TSL)~\cite{DBLP:conf/cav/Finkbeiner0PS19}.
In the finite-state setting, AIGER~\cite{biere_aiger_2007} also supports
registers (latches).
While arguably intuitive, this approach introduces a conceptual divide between the automaton's (or the formula's)
alphabet and its memory, resulting in intricate formal semantics.
For instance, register reset logic needs to be specified explicitly (e.g.,~\cite{biere_aiger_2011}).
RA can be constructed from Freeze LTL~\cite{DBLP:conf/lics/DemriL06}, which extends LTL~\cite{DBLP:conf/focs/Pnueli77} with
quantifiers to store and recall values in registers:
this implies that several decision problems for RA over infinite words are either
extremely complex or undecidable.
Similarly, satisfiability of TSL modulo theories is based on \emph{B\"uchi stream automata},
similar to RA but with support for rich expressions in updates,
and is also proven undecidable~\cite{DBLP:conf/fossacs/FinkbeinerHP22}.
The language for the Issy synthesis tool~\cite{DBLP:conf/cav/HeimD25}
supports game structures with next-state constraints that resemble obligations,
but it does not feature semantics for alternation.

\newcommand*{\Preds}{\mathcal{P}\!r}

\section{Obligation automata}\label{sec:oblaut}

\noindent\emph{Preliminary definitions.}
A \emph{theory} $T_V$ is a fixed interpretation of
\emph{terms} and \emph{predicates}, constructed from
\emph{variables} (taken from a set $V$) and \emph{constants}.
Let $\mathcal{T}(V)$ be the set of terms of a theory,
with free variables in $V$, and let
$\Preds(V)$ the set of predicates over $\mathcal{T}(V)$.
A partial (r. total) \emph{valuation} is a partial (r. total) mapping from variables to constant terms.
We denote by $\SEM{z}{v}$ the evaluation within $T_V$ of a term or predicate $z$
against a total valuation $v$, and by $\Vals$ the set of all total valuations.
When $p$ is a predicate, $v \models p$ is shorthand for $\SEM{p}{v} = \mathit{true}$.

An \emph{obligation} is a partial, possibly empty mapping from $V$ to terms in $\mathcal{T}(V)$.
For an obligation $o = \{x_1 \mapsto t_1, \ldots, x_n \mapsto t_n \}$,
with a slight abuse of notation
we will use $\SEM{o}{v}$ to denote the
predicate $\bigwedge_i (x_i = \SEM{t_i}{v})$.
For the empty obligation $\varepsilon$ we define $\SEM{\varepsilon}{v} = \mathit{true}$,
for every $v$.

\noindent\emph{Obligation automata with existential branching.}
An obligation automaton (modulo $T_V$) is a tuple
$\mathcal{A} = \langle V, Q, R, I, F, \mathit{Acc} \rangle$,
with $V$ the set of variables,
$Q$ a finite set of \emph{states},
$R$ an edge relation,
$I \subseteq Q$ a set of \emph{initial states},
$F \subseteq 2^R$ a set of (transition-based) \emph{acceptance sets},
and $\mathit{Acc}$ an \emph{acceptance condition}.
We assume that sets in $F$ are non-empty and
and numbered $F_1, \ldots, F_n$.

$R$ is a set of tuples $(q, \gamma, o, q')$
where $q, q' \in Q$ are respectively the edge's \emph{source} and \emph{target},
$\gamma \in \Preds(V)$ is its \emph{guard},
and $o \in (V \hookrightarrow \mathcal{T}(V))$ its \emph{obligation}.
Intuitively, such an edge indicates that,
when the automaton is in state $q$,
it can read a valuation $v$ such that
$v \models \gamma$ and then transition to state $q'$.
In this state, the additional constraint $\SEM{o}{v}$ will have to hold against the \emph{next} valuation read by $\mathcal{A}$.
We may also denote an edge by $q\xrightarrow{\gamma, o} q'$,
and use $R_q$ to denote the set of edges leaving $q$.

An Emerson-Lei \emph{acceptance condition} over acceptance sets $F$
is a term from the following grammar~\cite{DBLP:journals/scp/EmersonL87}:
\[
\zeta \Coloneqq
\mathit{true}
\mid \mathit{false}
\mid \mathrm{Fin}(k)
\mid \mathrm{Inf}(k)
\mid \zeta \land \zeta
\mid \zeta \lor \zeta,
\quad 1 \leq k \leq |F|
\]

Intuitively,
$\mathrm{Fin}(k)$ is only satisfied if the automaton
takes every transition from $F_k$ at most a finite number of
times; otherwise, $\mathrm{Inf}(k)$ is satisfied.
(As is commonplace for $\omega$-automata, whenever we mention 
that a state $q$ is ``accepting'', we merely mean that all edges in $R_q$
belong to some accepting set.)

\noindent\emph{Words, runs, acceptance, languages.}
A \emph{word} is a sequence of valuations $\mathbf{v} = v_1v_2\ldots \in \Vals^\omega$.
A \emph{run} over a word in an automaton $\mathcal{A}$, if it exists, is a sequence of edges
$e_i = (q_i, \gamma_i, o_i, q'_i)$, $i \in \mathbb{N}$,
such that $q_1 \in I$, and for every $i$:
\begin{enumerate*}
\item $q'_i = q_{i+1}$;
\item $v_i \models \gamma_i$; and
\item $v_{i+1} \models \SEM{o_i}v_i$.
\end{enumerate*}
Given an acceptance set $F_k$, a run $e_1e_2\ldots$ is \emph{accepting} for $\mathrm{Fin}(k)$ iff the set $\{ i \mid e_i \in F_k\}$ is finite;
otherwise, it is accepting for $\mathrm{Inf}(k)$.
A run is accepting for a conjunction (r. disjunction) of conditions iff it is accepting for all (r. at least one) of them.
Every run is accepting for $\mathit{true}$, and no run is accepting for $\mathit{false}$.
The \emph{language} of $\mathcal{A}$, $\mathcal{L}(\mathcal{A})$, is the set of words over which there exists an accepting run for $\mathit{Acc}$.
$\mathcal{A}$ is \emph{empty} iff $\mathcal{L}(\mathcal{A})$ is the empty set.

\noindent\emph{Product of obligations.}
Given two obligations $o_1, o_2$, we now want to introduce a
product operator $o_1 \times o_2$ to  formalise
the meaning of a transition that ``combines'' both obligations.
We will need this to formalise other concepts, including
the product of two automata and the semantics of universal branching.

When the domains of $o_1$ and $o_2$ are disjoint,
we can simply take the union of the two mappings;
however, in general the domains overlap.
In the simplest case,
we have an edge $(q, \gamma, o_1\times o_2, q')$
where $o_1, o_2$ map a single variable $x$ to two terms
$t_1$, $t_2$.
We choose to treat an edge like this as
equivalent to the edge $(q, \gamma \land t_1 = t_2, \{ x \mapsto t_1 \}, q')$.
That is, we reduce the obligation to $x \mapsto t_1$,
but also force the \emph{current} valuation to
satisfy $t_1 = t_2$, since this is the only way
to enforce that the \emph{next} symbol will
also obey $x \mapsto t_2$.
To generalise, every edge $(q, \gamma, o_1 \times o_2, q')$
is treated as the edge
$(q, \gamma \land \gamma', o', q')$,
where
\[
\begin{aligned}
\gamma' &= \bigwedge \{ t_{1j} = t_{2j} \mid  x_{ij} \in \dom(o_1) \cap \dom(o_2)\},& \text{with } i=1,2, \text{ and}\\
o' &= \{x_{ij} \mapsto t_{ij} \mid x_{ij} \in \dom(o_1) \cup (\dom(o_2)\setminus \dom(o_1)) \},& \text{with } i=1,2.
\end{aligned}
\]

This definition of product is not commutative:
$o_1 \times o_2$ will preserve terms from $o_1$ in the obligation and
relegate terms from $o_2$ to $\gamma'$, and vice versa for
$o_2 \times o_1$. However, by symmetry and transitivity of~$=$,
an edge with obligation
$o_1 \times o_2$ is semantically equivalent to one with $o_2 \times o_1$;
thus, this definition may be readily generalised to any number of obligations.

\noindent\emph{Product and sum of automata.}
Let us consider two automata $\mathcal{A}_{1,2}$ modulo the same theory.
Then, their product $\mathcal{A}_1 \otimes \mathcal{A}_2$
is an automaton with states in $Q_1 \times Q_2$ and initial states
$I_1 \times I_2$.
Whenever each automaton $\mathcal{A}_i$
has an edge $e_i$ from $q_i$ to $q_i'$
with guard $\gamma_i$ and obligation $o_i$,
then the product automaton
has an edge $e_1 \times e_2 = (q_1, q_2) \xrightarrow{\gamma_1 \land \gamma_2, o_1 \times o_2} (q'_1, q'_2)$.

To obtain the acceptance condition, let
$F^1_1,\ldots, F^1_n ,F^2_1, \ldots, F^2_m$ be the accepting sets of $\mathcal{A}_1, \mathcal{A}_2$.
Define new sets $F_1, \ldots, F_{n+m}$, and let
$e_1\times e_2 \in F_i$ iff either $e_1 \in F^1_i$ for $i \leq n$,
or $e_2 \in F^{2}_{i-n}$ for $n < i \leq n+m$.
Lastly, lift $\mathit{Acc}_1, \mathit{Acc_2}$ to $F$ by substituting references to
sets in $F^1, F^2$ with references to the corresponding sets in $F$.
Then, the acceptance condition of the product is
the conjunction of the (lifted) conditions.

The sum of two automata $\mathcal{A}_1 \oplus \mathcal{A}_2$
has states $Q_1 \cup Q_2$, initial states $I_1 \cup I_2$,
edges $R_1 \cup R_2$,
and acceptance condition $\mathit{Acc}_1 \lor \mathit{Acc_2}$.

\begin{thmrep}\label{thm:intersect}
Let $\mathcal{A}_1, \mathcal{A}_2$ be two B\"uchi obligation automata (BOA)
with languages $\mathcal{L}_1$ and $\mathcal{L}_2$.
Then $\mathcal{L}(\mathcal{A}_1 \otimes \mathcal{A}_2) = \mathcal{L}_1 \cap \mathcal{L}_2$.
\end{thmrep}
\begin{inlineproof}
Given in the appendix; simple extension of the proof for B\"uchi $\omega$-automata.
\end{inlineproof}
\begin{proof}
First, we prove that
$\mathcal{L}_1 \cap \mathcal{L}_2 \subseteq \mathcal{L}(\mathcal{A}_1 \otimes \mathcal{A}_2)$.
Let $\mathbf{v} \in \mathcal{L}_1 \cap \mathcal{L}_2$.
Then there is an accepting run $\rho_j =  e_{j1}e_{j2}\ldots$ in each $\mathcal{A}_j$.
Let $e_{ji} = (q_{ji}, \gamma_{ji}, o_{ji}, q'_{ji})$.
By construction, the product contains edges of the form
$e_i = e_{1i} \times e_{2i} = (q_{1i}, q_{2i}) \xrightarrow{\gamma, o_i} (q'_{1i}, q'_{2i})$,
with $\gamma = \gamma_{1i} \land \gamma_{2i} \land \gamma'$ and $o_i$ follow the
definition of $o_{1i} \times o_{2i}$.

For every $i$, we have that $v_i \models \gamma_{1i} \land \gamma_{2i}$:
otherwise, either $\rho_1$ or $\rho_2$ would not exist.
$v_i \models \gamma'$ because otherwise $v_{i+1}$ would not satisfy
$\SEM{o_{2i}}{v_i}$, and $\rho_2$ would not exist.
Finally, $o_i$ only includes constraints from $o_{1i}$ and $o_{2i}$.
Since $v_{i+1} \models \SEM{o_{1i}}{v_i}$ and $v_{i+1} \models \SEM{o_{2i}}{v_i}$,
we can conclude that $v_{i+1} \models \SEM{o_{i}}{v_i}$.
Thus, $\rho = e_1e_2\ldots$ is a run for $\mathbf{v}$ in the product.
This run contains infinitely many edges $(e_{1i} \times e_{2i})$
with $e_{1i}$ being a member of some accepting set in $\mathcal{A}_1$,
and infinitely many edges with $e_{2i}$ belonging to an accepting set in
$\mathcal{A}_2$, because $\rho_1$ and $\rho_2$ are both accepting.
But then, by construction of the acceptance condition of the product,
$\rho$ is accepting.

The proof for
$\mathcal{L}_1 \cap \mathcal{L}_2 \supseteq \mathcal{L}(\mathcal{A}_1 \otimes \mathcal{A}_2)$
is quite similar.
An accepting run $\rho$ in the product must correspond to two runs $\rho_1, \rho_2$
in $\mathcal{A}_1, \mathcal{A}_2$: intuitively, whenever some $v_i$
satisfies a conjunction of guards, or a product of obligations, it must
satisfy every individual guard or obligation. Therefore, $\rho_1, \rho_2$
must exist.
Are they accepting? Assume it is not the case: $\rho$ is accepting
but some $\rho_j$ is not. This means that, for every accepting set in
$\mathcal{A}_j$, $\rho$ takes every edge in that set only a finite number
of times.
However, by construction of the product's acceptance condition, this would also prevent
$\rho$ from being an accepting run in the product. Therefore, $\rho_1, \rho_2$
must be accepting in $\mathcal{A}_1, \mathcal{A}_2$.
\qed
\end{proof}

\noindent\emph{Semantics of universal branching.}
We now outline the (transition-based) semantics of automata with universal branching.
This is important as it allows one to define alternating automata, which are exponentially
more succinct than equivalent nondeterministic automata~\cite{DBLP:conf/icalp/BokerKR10}.
At any time, an automaton with universal branching may be in
a non-empty conjunction of states (or \emph{and-state}, for short);
the initial state may be an and-state as well.
Therefore, we amend the previous definition of an automaton $\mathcal{A}$ as follows:
the set $I$ is now a non-empty set of and-states,
and edges have the form $(q, \gamma, o, Z)$, where $Z$ is an and-state.

A run over a word $\mathbf{v} = v_1v_2\ldots \in \Vals^\omega$,
if it exists,
is an infinite directed acyclic graph (DAG) whose nodes are labelled
by and-states, and whose (unique) root is labelled by an element of $I$.
Note that multiple nodes may have the same and-state as their label.
Each arc in the tree has a label $(\gamma, o, \mathbf{f})$, 
where $\gamma$ is a predicate,
$o$ an obligation, and
$\mathbf{f}$ a set of indices referring to accepting sets:
we call these the \emph{accepting indices} of the arc.

The arcs in a run must satisfy additional constraints.
For a node labelled by $Z = \{q_1, \ldots, q_n\}$,
an arc labelled $(\gamma, o, \mathbf{f})$
from this node to one labelled by $Z'$ exists if and only if:
\begin{enumerate*}
\item There exist $n$ edges $e_j = (q_j, \gamma_j, o_j, Z'_j)$ such that $e_j \in R_{q_j}$ for every $j$;
\item $\gamma = \bigwedge_j \gamma_j$;
\item $o = o_1 \times \ldots \times o_n$;
\item $Z' = \bigcup_j Z'_j$;
\item $\mathbf{f} = \{ k \mid \exists j . e_j \in F_{k} \}$, where $F_{k}$ is the $k$-th accepting set of the automaton;
\item $v_i \models \gamma_i$, for every $i$;
\item $v_{i+1}\models \SEM{o_i}{v_i}$, for every $i$.
\end{enumerate*}

Intuitively, an alternating automaton runs multiple ``instances'' of
the same automaton at once, and the conditions above ensure that the run only
captures executions where all instances can read a given word.
In detail, condition 1 states that every instance must take a suitable edge to read the symbol.
Conditions 2--4 state that the arc composes edge guards by conjunction and obligations by product
(assuming that overlapping obligation domains are handled by extending the guard,
as seen above),
and that the target's label is the union of all end-states in the edges.
Condition 5 ensures that arcs hold the necessary information to determine acceptance.
Conditions 6 and 7 implement the semantics of guards and obligations.

A \emph{branch} of a run is a maximal path starting from its root.
For any branch, let $\mathbf{f}_1\mathbf{f}_2\ldots$ be the sequence of accepting indices on its arcs.
Then, the branch is \emph{accepting} for $\mathrm{Fin}(k)$ iff the set $\{ i \mid k \in \mathbf{f}_i\}$ is finite;
otherwise, the branch is accepting for $\mathrm{Inf}(k)$.
Conjunctions and disjunctions work as in the existential-branching case.
A run is accepting for an acceptance condition iff all its branches are accepting.

\smallskip\noindent\emph{Notes on expressiveness.}
We now state a few claims about the expressiveness of
obligation automata.
We will first formalise that every obligation automaton accepts a fragment of
the language of a corresponding $\omega$-automaton.
Then, we show that finite obligation automata go beyond $\omega$-regular
languages.
These claims may appear contradictory, but intuitively
the additional expressive power of obligation automata comes from
obligations within cycles, which effectively turn some variables
into a sort of memory cell.

\begin{thm}\label{thm:fragment}
Let $\mathcal{A}$ be an obligation automaton, and
$\hat{\mathcal{A}}$ the automaton we obtain from $\mathcal{A}$ by replacing every
obligation with $\varepsilon$. Then:
\begin{enumerate*}
    \item $\hat{\mathcal{A}}$ is a symbolic $\omega$-automaton over the same theory of $\mathcal{A}$;
    \item $\mathcal{L}(\mathcal{A}) \subseteq \mathcal{L}(\hat{\mathcal{A}})$.
\end{enumerate*}
\end{thm}
\begin{proof}
The first part of the theorem states that obligations are the
\emph{only} element that differentiates obligation automata from ordinary
symbolic automata: in fact, if we replace every obligation by the empty one,
the definitions of accepting run under
existential/universal branching from above collapse into
the traditional definitions for symbolic $\omega$-automata.
The second part of the theorem is justified by observing
that obligations may only ever
\emph{restrict} the language of an automaton.
In fact, by adding an obligation to any edge of $\hat{\mathcal{A}}$ we can
only ever \emph{reject} some valuations that would have been \emph{read} by
$\hat{\mathcal{A}}$, but not vice versa.
\qed
\end{proof}

\begin{thm}\label{thm:finite}
There exists at least one finite obligation automaton $\mathcal{A}$
that recognises a language for which no finite symbolic $\omega$-automaton exists.
\end{thm}
\begin{proof}

Consider the theory of linear integer arithmetic
(LIA) over a single integer variable $x$.
We may think of the language $\mathcal{L}$
of valuation sequences where $x$
repeatedly increments by one unit; i.e.,
$\mathcal{L}$ is made of sequences $v_0v_1\ldots$
such that $v_i(x) = v_0(x) + i$ for every $i$.
This language cannot be captured by any \emph{finite} symbolic
automaton. In fact, handling the first valuation
already requires an infinite set of transitions.
Symbolic predicates will not help us here:
if an initial transition can be triggered by multiple values of $x$,
then in the target state we inevitably lose information on the actual
valuation that had just been read, which makes it impossible to recognise
$\mathcal{L}$.
At the same time, the automaton must also have infinitely many states.
If that is not the case, by the pigeonhole principle the automaton must send
the prefixes of two different words into the same state.
This would make it impossible for the automaton to recognise the suffixes of
both words.
The best that we can do is an automaton with infinitely many states and
transitions, with the (co-B\"uchi) condition that no state is visited
infinitely often (Fig.~\ref{fig:sequence1}).
This also means that $\mathcal{L}$ is not $\omega$-regular.

By contrast, the B\"uchi obligation automaton (BOA)
of Fig.~\ref{fig:sequence2}
is able to recognise~$\mathcal{L}$.
We can prove this inductively. In the initial state, the automaton
may read any valuation $v_0$.
Every such valuation is the beginning of a word in $\mathcal{L}$,
therefore we recognise $\mathcal{L}$ at least up to the first symbol.
Now, assume we recognise every word up to the first $k$
symbols $v_0\ldots v_{k-1}$.
After $v_{k-1}$ is read and the automaton's only transition is taken,
the obligation yields the constraint
$v_k \models \SEM{x \mapsto x+1}{v_{k-1}}$, that is,
$v_k \models (x = v_{k-1}(x)+1)$.
But this means that we also read the $k+1$-th
symbol, and since we reject everything except that symbol,
we recognise $\mathcal{L}$ up to $k+1$ symbols.
Therefore, since the automaton recognises all of
$\mathcal{L}$ and rejects every word not in $\mathcal{L}$,
its language is $\mathcal{L}$.
\qed
\end{proof}

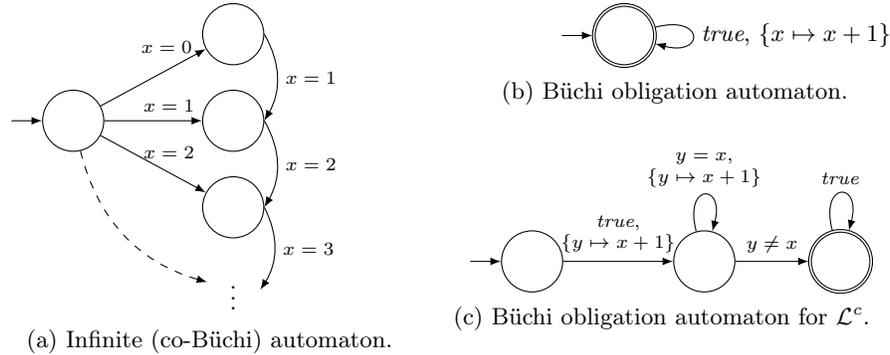
\begin{figure}[t]
\centering
\begin{minipage}[t]{0.49\textwidth}
\subfloat[Infinite (co-B\"uchi) automaton.\label{fig:sequence1}]{%
\centering
\begin{tikzpicture}[auto, every initial by arrow/.style={>=latex}]
\node [state,initial,initial text=] (q0){};
\node [state,right= 4em of q0] (q2){} ;
\node [state,above=1em of q2] (q1){} ;
\node [state,below=1em of q2] (q3){} ;
\node [circle,inner sep=2pt, below=1em of q3] (q4) {$\vdots$};
\node[right=6em of q2] (space) {};

\begin{scope}[->, >=latex]
\draw (q0) -- (q1) node[midway,xshift=18pt,yshift=6pt] {\scriptsize$x=0$};
\draw (q0) -- (q2) node[midway,xshift=6pt] {\scriptsize$x=1$};
\draw (q0) -- (q3) node[midway,above,xshift=6pt,yshift=-1pt] {\scriptsize$x=2$};
\draw[dashed] (q0) to[bend right] (q4);

\draw (q1.east) to[bend left] node[midway] {\scriptsize$x=1$} (q2.east);
\draw (q2.east) to[bend left] node[midway] {\scriptsize$x=2$} (q3.east);
\draw (q3.east) to[bend left] node[midway] {\scriptsize$x=3$} (q4.east);
\end{scope}
\end{tikzpicture}}
\end{minipage}
\hfill
\begin{minipage}[t]{0.5\textwidth}
\subfloat[B\"uchi obligation automaton.\label{fig:sequence2}]{%
\centering
\begin{tikzpicture}[auto, every initial by arrow/.style={>=latex}]
\node [state, initial, accepting,initial text=] (q0){};
\node[left=5em of q0] (extra) {};
\begin{scope}[->, >=latex]
\draw (q0) to[loop right] node[midway] {$\textit{true}$, $\{x \mapsto x+1\}$} (q0);
\end{scope}
\end{tikzpicture}}
\smallskip

\subfloat[B\"uchi obligation automaton for $\mathcal{L}^c$.\label{fig:sequence3}]{%
\centering
\begin{tikzpicture}[auto, every initial by arrow/.style={>=latex}]
    \node [state, initial,initial text=] (q0){};
    \node [state, right=4.5em of q0] (q1){};
    \node [state, accepting, right=3em of q1] (q2){};
    \begin{scope}[->, >=latex]
        \draw (q0) to node[midway,align=center,font=\scriptsize] {$\textit{true}$,\\$\{y \mapsto x+1\}$} (q1);
        \draw (q1) to[loop above] node[midway,align=center,font=\scriptsize] {$y= x$,\\$\{y \mapsto x+1\}$} (q1);
        \draw (q1) to node[midway,align=center,font=\scriptsize] {$y\neq x$} (q2);
        \draw (q2) to[loop above] node[midway,align=center,font=\scriptsize] {$\mathit{true}$} (q2);
    \end{scope}
\end{tikzpicture}}
\end{minipage}
\caption{Recognising the language $\mathcal{L}$ of a repeatedly-incrementing counter.}
\end{figure}

\noindent\emph{Notes on complementation.}
Figure~\ref{fig:sequence3} recognises the complement of language $\mathcal{L}$
from Theorem~\ref{thm:finite} by introducing a fresh variable $y$.
This variable is continually constrained to ``store'' $x+1$, so that a
successful test $x\neq y$ is enough to recognise that the word belongs
to $\mathcal{L}^c$. Words in $\mathcal{L}$ are instead rejected, as they make the
automaton continuously loops in the non-accepting state in the middle.
We should stress that, although the role of $y$ is reminiscent
of that of a register, formally there is no difference between $x$ and $y$ in
terms of semantics.

We suspect that complementation without either recurring to alternation or changing the theory domain $V$ is
impossible even for simple automata such as this one. Obligations may only introduce
\emph{equality} constraints, and for complementation we would likely need a dedicated construct to
add \emph{inequality} constraints.
We sketch such a construct in Appendix~\ref{apx:nobl} and will investigate it in future work.

\newcommand*{\NONCE}{\mathcal{N}\kern-2.5pt\mathit{once}}
For other automata, complementation appears to
be significantly harder, and it is unclear whether any such extension would
help.
For instance, consider a more complex example from the literature~\cite{DBLP:conf/lics/DemriL06}:
the language $\NONCE_x$ of words such that $x$ never takes the same value
more than once.
Interestingly, while this language is relatively easy to capture with register
automata, it requires a rather complex obligation automaton. In turn,
its complement $\NONCE_x^c$ can be recognised by a simple nondeterministic BOA, but not by a
simple RA. Let us start by discussing this automaton (Fig.~\ref{fig:noncec}).
In the initial state, we admit every valuation, but constrain $y$
in the next valuation to assume the current value of $x$.
Then, as long as we read symbols with $x\neq y$, we may nondeterministically
decide to hold the current value of $y$, or replace it by the value of $x$.
As soon as the test $x=y$ holds, we transition into the accepting state.
Thus, words that belong to $\NONCE_x^c$ always have an accepting run associated
to them, unlike words from $\NONCE_x$ which force the automaton to remain in the middle
state forever.

Recognising $\NONCE_x$ is rather more difficult.
Here we need the intuition that, in order to reject a word
from $\NONCE_x^c$,
we need to ``store'' in $y$ the first occurrence of the
duplicated value, and then hold that value until the second occurrence
appears.
To capture this behaviour, we need to consider \emph{every} possible sequence of the
two obligations $y \mapsto x$, $y \mapsto y$,
by means of universal branching (Fig.~\ref{fig:nonce}).
We start, again, by setting obligation $y \mapsto x$ after reading any
initial valuation.
Then, we enter a gadget composed of two accepting states,
each connected to each other and itself by edges that are guarded by
$x\neq y$, and that set either $y \mapsto x$ (for edges labelled by $\lambda_x$) or
$y \mapsto y$ (for $\lambda_y$).
As soon as $x=y$ holds, we leave the gadget and reach a non-accepting sink.

Words that are in $\NONCE_x$ have no way of reaching this sink,
and remain in the accepting gadget. The intuition here is that the
specific order of obligations over $y$ \emph{does not matter} for these words: if every value
of $x$ is different from all others, and $y$ may only store a previous value of
$x$, then the test $x=y$ will fail no matter which concrete value is ``held''
in $y$.
On the contrary, for words in $\NONCE_x^c$ there exists at least one
execution (that is, a branch in its run tree) that performs the ``correct''
sequence of obligations which makes the test $x=y$ hold when a value is seen for
the second time. This branch is not accepting because it ends up in the
non-accepting sink state, and thus the entire run is not accepting.\footnote{See Figure~\ref{fig:nonceruns} in appendix for a graphical representation of an accepting and a non-accepting run for this automaton.}

\begin{figure}[t]
\subfloat[Automaton for $\NONCE_x^{c}$.\label{fig:noncec}]{%
\begin{tikzpicture}[auto, every initial by arrow/.style={>=latex}]
\node [state, initial above, initial text=] (q0) {};
\node [state, right=2.5em of q0] (q1) {};
\node [state, accepting, right=2.5em of q1] (q2) {};

\begin{scope}[->, >=latex]
\draw (q0) to node[midway,font=\scriptsize,align=center] {$\mathit{true}$\\$y \mapsto x$} (q1);
\draw (q1) to node[midway,font=\scriptsize,align=center] {$x=y$} (q2);
\draw (q1) to[loop above] node[midway,font=\scriptsize,align=center] {$x\neq y,\;y \mapsto x$} (q1);
\draw (q1) to[loop below] node[midway,font=\scriptsize,align=center] {$x\neq y,\;y \mapsto y$} (q1);
\draw (q2) to[loop above] node[midway,font=\scriptsize,align=center] {$\mathit{true}$} (q2);
\end{scope}
\end{tikzpicture}}
\hfill
\subfloat[Alternating automaton for $\NONCE_x$. Let $\lambda_x = (x\neq y, y \mapsto x)$, $\lambda_y = (x\neq y, y \mapsto y)$.\label{fig:nonce}]{%
\begin{tikzpicture}[auto, every initial by arrow/.style={>=latex}]
\node [state, initial above, initial text=] (q0) {$q_0$};
\node [state, accepting, right=2.5em of q0] (q1) {$q_1$};
\node [state, accepting, right=4em of q1] (q2) {$q_2$};
\node [state, right=4em of q2] (q3) {$q_3$};

\node [yshift=4pt,circle, inner sep=1pt, fill, right=2em of q1,label={[font=\scriptsize]above:$\lambda_x$}] (q12x) {};
\node [yshift=-4pt,circle, inner sep=1pt, fill, right=2em of q1,label={[font=\scriptsize]below:$\lambda_y$}] (q12y) {};

\node [circle, inner sep=1pt, fill, above=1.2em of q2] (q21x) {};
\node [circle, inner sep=1pt, fill, below=1.2em of q2] (q21y) {};

\begin{scope}[->, >=latex]
\draw (q0) to node[midway,font=\scriptsize,align=center] {$\mathit{true}$\\$y\mapsto x$} (q1);
\draw (q2) to node[midway,font=\scriptsize,align=center] {$x=y$} (q3);

\draw (q1) -- ++(0,-3.5em) -| node[font=\scriptsize,pos=0.37] {$x=y$} (q3);
\draw (q12x) -- (q2.west |- q12x);
\draw (q12x) to[bend right] (q1);
\draw (q12y) -- (q2.west |- q12y);
\draw (q12y) to[bend left] (q1);
\draw (q21x) to[bend right] (q1);
\draw (q21y) to[bend left] (q1);
\draw (q21x) to[bend right=400] (q2);
\draw (q21y) to[bend left=400] (q2);

\draw (q3) to[loop above] node[midway,font=\scriptsize,align=center] {$\mathit{true}$} (q3);
\end{scope}
\draw (q1) -- (q1.east |- q12x) -- (q12x);
\draw (q1) -- (q1.east |- q12y) -- (q12y);

\draw (q2) to node[right]{\scriptsize$\lambda_x$} (q21x);
\draw (q2) to node[right]{\scriptsize$\lambda_y$} (q21y);

\end{tikzpicture}}
\caption{Automata for a language of non-repeating values of $x$.}

\end{figure}
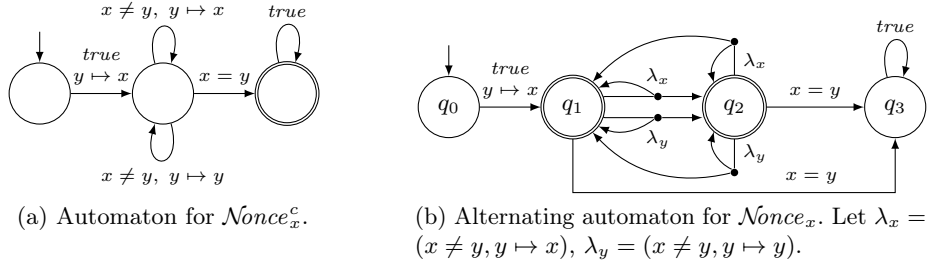

\newcommand*{\ObLTL}{\mathrm{ObLTL}}

\section{The \HOApp format}\label{sec:hoapp}

In this section, we give an overview of our proposed HOA dialect.
We first briefly review the original HOA format\footnote{
This overview focusses
on automata with transition-based acceptance and explicit edge labels, due to space constraints.
The format is rather more flexible than this:
for more details, we refer the reader to \HOAone's full documentation (\url{https://adl.github.io/hoaf/}).}
(Section~\ref{sec:hoa}),
and then describe our syntactic additions to it (Section~\ref{sec:syntax}).
Notice that these changes are meant to be incremental: any \HOAone automaton should also be accepted by a \HOApp parser.
Then, we introduce a simple type system for the new format (Section~\ref{sec:types}), and discuss its semantics (Section~\ref{sec:sem}).

\subsection{Overview of \HOAone}\label{sec:hoa}

The description of an automaton in \HOAone is split into two parts: a \emph{header} containing general information, and a \emph{body} describing its states and transitions.

The header always starts with \texttt{HOA: v1} and is organised into header items (or simply \emph{items}), which are made of a name, a colon sign, and some data. For instance, the item \texttt{States: 10} declares that the automaton has 10 states, while \texttt{Start: 0} indicates that the automaton starts from the state with index 0.
The \texttt{AP} item introduces a number of \emph{atomic propositions} ($AP$). For instance, \texttt{AP: 2 "x" "y"} declares that the automaton's edges are labelled by predicates over two propositions $x,y$.
The \texttt{Acceptance} item contains the automaton's (Emerson-Lei) acceptance condition,
following the same grammar described in Section~\ref{sec:oblaut}.
For instance, a valid acceptance item would be \texttt{Acceptance: 2 Fin(0) \& Inf(1)},
where the initial \texttt{2} indicates the total number of acceptance sets in the automaton.
(The actual composition of these sets is defined in the body of the automaton).
Lastly, \emph{aliases} are another kind of header item worth mentioning.
An alias is an identifier prefixed by \texttt{@}, and associated to an expression.
When used in an expression, an alias is expanded into its corresponding expression, akin to a syntactic macro.
Expansion is recursive, meaning that an alias may refer to other, simpler aliases.
Expressions refer to an atomic proposition simply by its (0-based) position in the \texttt{AP} header item. For instance, given \texttt{AP: 2 "x" "y"}, the item \texttt{Alias: @x-and-noty 0 \& !1} defines an alias for the predicate $x \land \neg y$.
Users may also attach custom items, as long as they start with a lowercase
letter: tools that do not recognise these items are allowed to simply ignore them.

The body
is introduced by the keyword \texttt{-{}-BODY-{}-} and
contains a sequence of state descriptors. Each of these starts with the \texttt{State:} keyword, followed by the \emph{index} of the state and zero or more \emph{edges} describing its outgoing transitions.
Each edge has a predicate over $AP$ (its \emph{guard}), the index of a \emph{target state}, and an optional acceptance \emph{signature}, i.e., the list of acceptance sets the edge belongs to (if any).
The body is terminated by the keyword \texttt{-{}-END-{}-}.

To illustrate the format, we provide an example of a simple \HOAone description of an automaton for the LTL property $G\, x$, along with its graphical representation (Fig.~\ref{fig:Gx}).
From initial state $s_0$, the automaton loops back to $s_0$ if $x$ holds, or moves to state $s_1$ otherwise.
State $s_1$ is a sink, as it only allows an always-enabled transition that loops back to $s_1$.
The edges leaving $s_0$ belong to acceptance set 0, and the B\"uchi condition $\mathrm{Inf(0)}$ states that the automaton must take at least one of them infinitely often.
This may only happen if $\neg x$ never holds, which indeed matches the semantics of $G\, x$.
The example also illustrates the usage of an alias (\texttt{@x}) to improve the readability of labels.

\begin{figure}[t]
\hspace{2em}
\begin{minipage}[t]{0.32\textwidth}
\strut\vspace*{-\baselineskip}\newline
\begin{lstlisting}[breaklines=true,basicstyle=\scriptsize\ttfamily,numbers=left]
HOA: v1
States: 2
Start: 0
AP: 1 "x"
Acceptance: 1 Inf(0)
Alias: @x 0
\end{lstlisting}
\vfill
\end{minipage}
\begin{minipage}[t]{0.22\textwidth}
\strut\vspace*{-\baselineskip}\newline
\begin{lstlisting}[breaklines=true,basicstyle=\scriptsize\ttfamily,numbers=left,firstnumber=7]
--BODY--
State: 0 "s0"
[@x] 0 {0}
[!@x] 1 {0}
State: 1 "s1"
[t] 1
--END--
\end{lstlisting}
\end{minipage}
\hfill
\begin{minipage}[t]{0.38\textwidth}
\strut\vspace*{-\baselineskip}\newline
\centering
\begin{tikzpicture}[auto, every initial by arrow/.style={>=latex}]
\begin{scope}[circle, inner sep=3pt]
\node [state,initial,initial text=] (s0) {\footnotesize$s_0$};
\node [state, right= 4em of s0] (s1) {\footnotesize$s_1$};
\end{scope}

\begin{scope}[->,>=latex]
\draw (s0) -- (s1) node[midway, label=below:{\SET}] {\footnotesize$\neg x$};
\draw (s0) to[loop above] node[midway, align=center] {{\footnotesize$x$} \SET} (s0);
\draw (s1) to[loop above] node[midway] {\footnotesize$\True$} (s1);

\end{scope}
\end{tikzpicture}
\end{minipage}

\caption{\HOAone description of a B\"uchi automaton for $G\, x$.}\label{fig:Gx}
\end{figure}
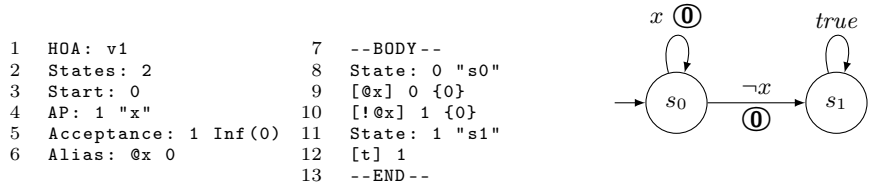

\subsection{Syntactic extensions to \HOAone}\label{sec:syntax}

First, we detail the new syntactic constructs that \HOApp
adds to \HOAone.
We use an extended BNF adapted from the latter's documentation:
a quoted string \texttt{"lit"} is a literal, an unquoted name \texttt{rule} refers to a production rule,
\texttt{(p)?} denotes an optional occurrence of \texttt{p},
and \texttt{(p)*} denotes zero or more occurrences of \texttt{p}.
Production rules for non-terminals and terminals are defined with \texttt{::=} and~\texttt{:}, respectively.
If we omit a production rule, it is implied that the rule is part of the \HOAone specification;
similarly, when we use ellipsis in a production rule (e.g., \texttt{p ::= p1 | ...}),
we are extending a \HOAone rule with the same name.

\smallskip\noindent\emph{The \texttt{HOA} item.}
A \HOApp specification must start with the \texttt{HOA: v1pp} item.

\smallskip\noindent\emph{The \texttt{AP} and \texttt{AP-type} items.}
The header item \texttt{AP} has the same syntax as in \HOAone.
However, we now refer to the elements of \texttt{AP} as \emph{variables},
rather than atomic propositions.
We introduce an optional header item \texttt{AP-type}:
if present, it shall be followed by $|AP|$ type names, which are $\Bools$, $\Ints$, or $\Reals$.
The $i$-th element of this vector gives the type of the $i$-th variable.
If this header is omitted, we shall assume all variables to be Booleans.

\begin{table}[t]
\centering
\caption{Extended syntax of \HOApp.}\label{fig:hoapp-syntax}
\hrulefill
{\footnotesize
\begin{align*}
\texttt{header-item } &\texttt{::= "AP-type:" TYPE*}
\texttt{ |} \texttt{ "controllable-AP:" INT INT* }\\
&\texttt{|} \texttt{ "assume:" ltl }
 \texttt{|} \texttt{ "guarantee:" ltl } \texttt{| ...}
\\
\texttt{ltl } &
 \texttt{::= label-expr }
 \texttt{|} \texttt{ "!" ltl | MOD ltl}\\
&\texttt{|} \texttt{ ltl "\&" ltl | ltl "|" ltl | ltl "U" ltl}\\
\texttt{label } &\texttt{::= } \texttt{"[" label-expr ("\$" obl)? "]"}\\
\texttt{label-expr } &\texttt{::= }
    \texttt{INTLIT | REALLIT} \texttt{ |} \texttt{ "-" label-expr}\\
    &\texttt{|} \texttt{ label-expr OP label-expr}
    \texttt{ |} \texttt{ label-expr CMP label-expr}\\
    &\texttt{|} \texttt{ label-expr TEST label-expr} \texttt{ |} \texttt{ ...}\\
\end{align*}
}\\[-4em]
\begin{minipage}[t]{0.45\textwidth}
\begin{align*}
    \texttt{obl } &
    \texttt{::= } \texttt{asgn ("," asgn)*}\\
    \texttt{asgn } &\texttt{::= }
    \texttt{lhs ":=" label-expr}\\
    \texttt{lhs } &\texttt{::= } \texttt{INT | ANAME}\\
    \texttt{TYPE } &\texttt{: "bool" | "int" | "real"}\\
    \texttt{MOD } &\texttt{: "F" | "G" | "X"}\\
\end{align*}
\end{minipage}\hfill
\begin{minipage}[t]{0.45\textwidth}
\begin{align*}
    \texttt{OP } &\texttt{: }
    \texttt{"+" | "-" | "*"}\\
    \texttt{CMP } &\texttt{: }
    \texttt{"<" | "<=" | ">" | ">="}\\
    \texttt{TEST } &\texttt{: }
    \texttt{"==" | "!="}\\
    \texttt{INTLIT } &\texttt{: } \texttt{i(0|[1-9][0-9]*)}\\
    \texttt{REALLIT } &\texttt{: } \texttt{r(0|[1-9][0-9]*).[0-9]*}
\end{align*}
\end{minipage}

\hrulefill%
\end{table}

\smallskip\noindent\emph{Extended expression syntax.}
We extend the grammar of expressions (\texttt{label-expr}) with
numeric constants, comparisons, and arithmetic operators.
Operators bind in the following, descending order of precedence:
\texttt{! -} (unary);
\texttt{*};
\texttt{+ -};
\verb+< <= > >=+;
\verb+== !=+;
\verb+&+;
\verb+|+.
All binary operators are left-associative; unary ones (\texttt{-} and \texttt{!}) are right-associative.

Like in \HOAone, plain numeric literals are references to variables;
thus, we use prefixes \texttt{i} and \texttt{r} to denote integer and real constants, respectively.

\smallskip\noindent\emph{Extended edge syntax.}
Edges are labelled by a predicate, which is the edge's \emph{guard},
and zero or more assignments that constitute the edge's \emph{obligation}.
The left-hand side of each assignment shall be a reference to a variable:
thus, either a plain integer or an alias that resolves to an integer.
Examples of (syntactically) valid labels in \HOApp are \texttt{[t \$ 0 := f]} or \texttt{[1 \$ 2 := 2 + 3, 0 := i1]}.

\noindent\emph{Additional header items.} Our format adopts the 
\texttt{controllable-AP} item~\cite{DBLP:journals/corr/abs-1912-05793} to specify
whether a variable is set adversarially by the environment or may be controlled
by the system, allowing \HOApp to represent game structures.
We also introduce additional items \texttt{assume}, \texttt{guarantee} that contain LTL formulas,
and may appear zero or more times.

\subsection{Typing rules}\label{sec:types}

Our format considers three data types, namely Booleans (\Bools),
integers (\Ints), and real numbers (\Reals),
with domains $\dom(\Bools) = \{\textit{true}, \textit{false}\}$,
$\dom(\Ints) = \Z$, $\dom(\Reals) = \R$.
We also introduce an internal type \LTLs of well-typed LTL formulas.
We assume $\Ints$ to be a subtype of $\Reals$ and $\Bools$ to be a subtype of $\LTLs$,
denoted by $\Ints \leq \Reals$ and $\Bools \leq \LTLs$, respectively.
Subtyping is reflexive, transitive, and antisymmetric.
Given types $\tau_1, \tau_2$, we define $\max(\tau_1, \tau_2)$ to be $\tau_1$ iff $\tau_2 \leq \tau_1$; $\tau_2$ iff $\tau_1 \leq \tau_2$; and $\bot$ otherwise. We use $\max(\cdot)$ to describe some implicit type conversion rules: for instance, the result of $\texttt{i1} + \texttt{r2.5}$ has type $\max(\Ints, \Reals)$,
which in this case is $\Reals$;
$\texttt{i0} + \texttt{i2}$, in turn, has type $\Ints$.
We do not define a separate type for assignments; instead, we use the notation
$x \mathbin{\texttt{:=}} e : \diamond$ to denote that an assignment is \emph{well-typed}.

Given a \HOApp automaton with variables
$\ x_1, \ldots, x_n$ and
$\texttt{AP-type:}\ \tau_1, \ldots, \tau_n$, we define
$\textit{type}(i) = \tau_i$, for $i=1,\ldots,n$.
Under this definition, we
can type-check expressions and obligations according to the rules in Table~\ref{tab:typerules}.
We say that the automaton is \emph{well-typed} iff all the following requirements hold:
\begin{enumerate*}[label=(\roman*)]
    \item all its guards have type $\Bools$;
    \item all of its obligations
    are well-typed; and
    \item all formulas under LTL header items (\texttt{assume}, \texttt{guarantee}) have type $\LTLs$.
\end{enumerate*}

\begin{table}[t]
    \caption{Typing rules. Assume $\kappa$ an \texttt{INT} term and $\texttt{r}\rho$ a \texttt{REALLIT} term.}\label{tab:typerules}
\footnotesize
\hrulefill
\[
\arraycolsep=9pt\def\arraystretch{2}
\begin{array}{cccc}
\infer{
\def\arraystretch{0.6}\begin{array}{c}
\texttt{t} : \Bools\\
\texttt{f} : \Bools\\
\texttt{i}\kappa : \Ints\\
\texttt{r}\rho : \Reals
\end{array}}{} 
&
\infer{\kappa : \tau}{\textit{type}(\kappa) = \tau} &
\infer{e : \tau}{e : \sigma & \sigma \leq \tau} &
\infer{\texttt{-}e : \tau}{e : \tau & \tau\leq\Reals}
\end{array}
\]\\[-2em]
\[
\arraycolsep=9pt\def\arraystretch{2}
\begin{array}{ccc}
\infer{
\def\arraystretch{0.4}\begin{array}{c}
e_1 \texttt{ CMP } e_2 : \Bools\\
e_1 \texttt{ TEST } e_2 : \Bools\\
e_1 \texttt{ OP } e_2 : \max(\tau_1, \tau_2)\\
\end{array}
}{e_1 : \tau_1 & e_2: \tau_2 & \tau_1, \tau_2 \leq \Reals}
&
\infer{
\def\arraystretch{0.6}\begin{array}{c}
e_1 \texttt{ TEST } e_2 : \max(\tau_1, \tau_2)\\
e_1 \texttt{ \& } e_2 : \max(\tau_1, \tau_2)\\
e_1 \texttt{ | } e_2 : \max(\tau_1, \tau_2)\\
\end{array}
}{e_1 : \tau_1 & e_2: \tau_2 & \tau_1, \tau_2 \leq \LTLs}
&
\infer{%
\def\arraystretch{0.6}\begin{array}{c}
\texttt{X}\ \varphi : \LTLs\\
\texttt{F}\ \varphi : \LTLs\\
\texttt{G}\ \varphi : \LTLs\\
\end{array}%
}{\varphi : \LTLs}
\\[1em]
\infer{\texttt{!}e : \tau}{e : \tau & \tau \leq \LTLs}
&
\infer{%
\def\arraystretch{0.6}\begin{array}{c}
\varphi_1 \mathbin{\texttt{U}} \varphi_2: \LTLs\\
\end{array}%
}{\varphi_1 : \LTLs & \varphi_2 : \LTLs}
&
\infer{\textit{lhs} \mathbin{\texttt{:=}} e : \diamond}{\textit{lhs} : \tau & e : \tau}
\end{array}
\]

\hrulefill
\end{table}

\subsection{Semantics}\label{sec:sem}

A well-typed \HOApp specification describes an obligation automaton
$\mathcal{A} = \langle V, Q, R, I, F, \mathit{Acc} \rangle$.
The set of variables $V$ is given by the \texttt{AP:} item.
$Q, R, F$ are derived by the specification's body: each \texttt{State} term
defines a state $q \in Q$ and its associated $R_q$, as well as the
acceptance sets $F$ to which the edges belong (if any).
$I$ and $\mathit{Acc}$ are given by \texttt{Start:} and \texttt{Acceptance:}
items, respectively.

Notice that \texttt{State:} items may contain conjunctions of states;
similarly, the target of any edge may be a conjunction of states.
A \HOApp automaton features universal
branching if and only if any of these features occur.

\begin{table}[t]
\caption{Denotational semantics of expressions.}\label{tab:sem}
\hrulefill
\footnotesize
\begin{align*}
    \SEM{\texttt{t}}{\Val} &= \True &
    \SEM{\texttt{i}\kappa}{\Val} &= \kappa &
    \SEM{!e}\Val &= \neg\SEM{e}\Val &
    \SEM{\texttt{f}}{\Val} &= \False &
    \SEM{\texttt{r}\rho}{\Val} &= \rho &
    \SEM{\texttt{-}e}\Val &= -\SEM{e}\Val
\end{align*}\\[-3.8em]
\begin{align*}
    \SEM{\kappa}{\Val} &= \Val(p_\kappa) &
    \SEM{e_1 \mathbin{\texttt{+}}  e_2}\Val &=
    \SEM{e_1}\Val + \SEM{e_2}\Val &
    \SEM{e_1 \mathbin{\texttt{\&}} e_2}\Val &=
    \SEM{e_1}\Val \land \SEM{e_2}\Val\\
    \SEM{e_1 \mathbin{\texttt{*}} e_2}\Val &=
    \SEM{e_1}\Val \cdot \SEM{e_2}\Val &
    \SEM{e_1 \mathbin{\texttt{-}} e_2}\Val &=
    \SEM{e_1}\Val - \SEM{e_2}\Val &
    \SEM{e_1 \mathbin{\texttt{|}} e_2}\Val &=
    \SEM{e_1}\Val \lor \SEM{e_2}\Val
    \end{align*}\\[-3.8em]
    \begin{align*}
    \SEM{e_1 \texttt{ < } e_2}\Val &=
    \True  \text{ iff } \SEM{e_1}\Val < \SEM{e_2}\Val%
    \quad
    &
    \quad
    \SEM{e_1 \texttt{ == } e_2}\Val &=
    \True  \text{ iff } \SEM{e_1}\Val = \SEM{e_2}\Val%
\end{align*}

\hrulefill
\end{table}

\noindent\emph{Semantics of \texttt{guarantee}.}
Each \texttt{guarantee} header item specifies an additional acceptance condition in the form
of an LTL formula. Atoms in the formula may be arbitrary predicates over the automaton's variables.
The intended meaning of these items is that the automaton accepts only those words that satisfy its acceptance
condition \emph{and} the additional formulas.

Formally, let $\mathcal{B}(\varphi)$ be the B\"uchi automaton for LTL property $\varphi$.
Such an automaton may be obtained by treating the atoms of $\varphi$ as
atomic propositions, and then using the usual construction for LTL~\cite{DBLP:conf/lics/VardiW86}.
An implementation of this procedure is presented in Section~\ref{sec:tool}.
Then, a \HOApp specification
with one or more \texttt{guarantee} formulas
$\varphi_1, \varphi_2, \ldots$ describes the automaton
$\mathcal{A} \otimes \mathcal{B}(\varphi_1 \land \varphi_2 \land \ldots)$,
where $\mathcal{A}$ is the automaton given by the specification ignoring LTL header items.

\noindent\emph{Semantics of \texttt{assume}.}
An automaton with an \texttt{assume} formula
immediately accepts every run where the assumption is \emph{not} satisfied.
Thus, a specification with \texttt{assume} formulas $\alpha_1, \alpha_2, \ldots$
and guarantees $\varphi_1, \varphi_2, \ldots$
describes the automaton
$\mathcal{B}(\neg(\alpha_1 \land \alpha_2 \land \ldots)) \oplus (\mathcal{A} \otimes \mathcal{B}(\varphi_1 \land \varphi_2 \land \ldots))$,
where $\mathcal{A}$, $\mathcal{B}(\cdot)$ are as defined above.

\begin{figure}[t]
\centering
\begin{minipage}[t]{0.45\textwidth}
\begin{lstlisting}[
caption={Example: Simple arbiter (after auto-aliasing).},
label={ex:arbiter},
breaklines=true,
basicstyle=\scriptsize\ttfamily
]
HOA: v1pp
States: 3
Start: 0
AP: 4 "x" "dec" "y" "pause"
AP-type: int bool int bool
controllable-AP: 0 1
Acceptance: 1 Inf(0)
assume: G F !@pause
Alias: @x 0
Alias: @dec 1
Alias: @y 2
Alias: @pause 3
--BODY--
State: 0 "s0" {0}
[@x = i0 & @y > i0 $ @x := @y] 1
[@x != i0 $ @x := @x] 2
State: 1 "s1"
[@dec & !@pause $ @x := @x - i1] 1
[!@dec & !@pause $ @x := @x] 0
[@pause $ @x := @x] 1
State: 2 "s2"
[t $ @x := @x] 2 "s2"
--END--
\end{lstlisting}
\end{minipage}
\hfill
\begin{minipage}[t]{0.52\textwidth}
\begin{lstlisting}[
caption={The arbiter from Listing~\ref{ex:arbiter}, lowered into \HOAone.},
label={ex:arbiter-v1},
breaklines=true,
basicstyle=\scriptsize\ttfamily
]
HOA: v1
States: 3
Start: 0
AP: 11 "@dec" "@pause" "@x" "@x != i0"
    "@x - i1" "@x := @x" "@x := @x - i1"
    "@x := @y" "@x == i0" "@y" "@y > i0"
acc-name: Buchi
Acceptance: 1 Inf(0)
v1pp-AP: 4 x dec y pause
v1pp-AP-type: int bool int bool
v1pp-assume: "GF!@pause"
v1pp-controllable-AP: 0 1
--BODY--
State: 0 "s0" {0}
[7&8&10] 1
[3&5] 2
State: 1 "s1"
[0&!1&6] 1
[1] 1
[!0&!1&5] 0
State: 2 "s2"
[5] 2
--END--
\end{lstlisting}
\end{minipage}

\begin{tikzpicture}[auto, every initial by arrow/.style={>=latex}]
\node [state,accepting,initial above,initial text=] (s0) {\footnotesize$s_0$};
\node [state, right= 11em of s0] (s1) {\footnotesize$s_1$};
\node [state, left= 7em of s0] (s2) {\footnotesize$s_2$};

\begin{scope}[->,>=latex]
\draw (s0) -- (s1) node[midway] {\footnotesize$x = 0 \land y > 0, \{x \mapsto y\}$};
\draw (s0) -- (s2)  node[midway, above] {\footnotesize$x \neq 0, \{x \mapsto x\}$};
\draw (s1) to[bend left] node[midway] {\footnotesize$\neg\mathit{dec}\land\neg\mathit{pause}, \{x \mapsto x\}$} (s0);

\draw (s1) to[loop above] node[midway] {\footnotesize$\mathit{dec}\land\neg\mathit{pause}, \{x \mapsto x-1\}$} (s1);
\draw (s1) to[loop right] node[midway] {\footnotesize$\mathit{pause}, \{x \mapsto x\}$} (s1);
\draw (s2) to[loop above] node[midway] {\footnotesize$\True, \{x \mapsto x\}$} (s2);

\end{scope}
\end{tikzpicture}
\caption{Graphical representation of the automaton from Listing~\ref{ex:arbiter}.}\label{fig:arbiter}
\end{figure}
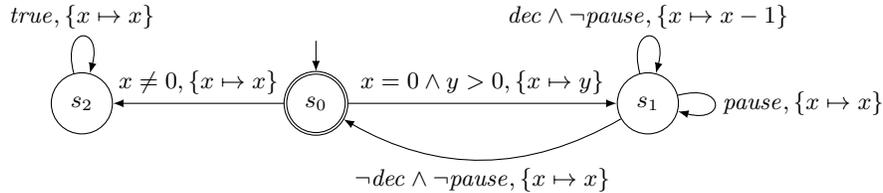

\subsection{Example: A simple arbiter}
In this section, we illustrate \HOApp's features by showing how it may be used to describe games over infinite-state arenas,
similarly to how~\cite{DBLP:journals/corr/abs-1912-05793} uses (an extension of) \HOAone as a formalism for
finite-state reactive synthesis problems.

Namely, we introduce an automaton (Listing~\ref{ex:arbiter} and Fig.~\ref{fig:arbiter})
for a game where the environment controls an integer \texttt{y} and a Boolean \texttt{pause}, while the other
variables (\texttt{x} of type $\Ints$, \verb|dec| of type $\Bools$) are left to the controller.
Naturally, both players still have to obey the automaton's obligations when giving values to these variables.
In the initial state $s_0$, the environment picks a value for \texttt{y}, which is then assigned to \texttt{x} as we move to state $s_1$.
Here, the controller can manipulate \verb|dec| to choose whether to decrement \texttt{x} and remain in $s_1$, or do nothing (keeping \texttt{x} as-is) and return to $s_0$.
When the controller chooses the former, the environment can prevent the decrement by setting \texttt{pause} to true.
However, the assumption $ G F\, \neg\texttt{pause}$ mandates that the environment must allow \texttt{x} to decrement infinitely often.
Whenever we are in $s_0$ and \texttt{x} is not equal to 0, the game goes to a sink state $s_2$.
The goal of the controller is to ensure that 
the automaton takes infinitely often any of the transitions leaving $s_0$.
Therefore, every time $s_1$ is visited, the controller must precisely decrement \texttt{x} to 0
(fighting the environment's delays while doing so) before going back to $s_0$: otherwise, the automaton will reach $s_2$, where the acceptance condition is impossible to satisfy and the environment wins.

\tikzset{
        doc/.style={
            document,
            draw=black!80,
            thick,
            minimum width=6mm,
            minimum height=8.48mm,
            font=\ttfamily\Large
        }
    }
\begin{figure}[t]
    \centering
\resizebox{0.9\textwidth}{!}{%
\begin{tikzpicture}[auto]
  \node[] (center) {};
  \node[doc, above=0.6em of center, document dog ear=4pt, label={below:\scriptsize\HOApp}] (a) {$\mathcal{A}$};
  \node[doc, below=0.6em of center, document dog ear=4pt, label={below:\scriptsize\HOApp}] (b) {$\mathcal{B}$};
  \node[below=2em of b, inner sep=0pt, label={below:\scriptsize{LTL}}] (phi) {\Large$\varphi$};

  \node[above=0.5em of a,inner sep=0pt] (args) {\texttt{args}};
  \node[below=2em of phi,inner sep=0pt] (types) {\texttt{types}};

  \node[draw, minimum height=7.2em, right=4em of center,fill=white,yshift=0.4em] (frontend) {\rotatebox{90}{Front end}};
  \node[draw, right=4em of center, minimum height=4em, align=center, fill=white, yshift=-5.7em] (ltlfrontend) {\rotatebox{90}{\parbox{4em}{\centering\scriptsize LTL\\front end}}};
  \node[draw, minimum height=8em, right=9em of center,fill=white] (lower) {\rotatebox{90}{\scriptsize Lower to \HOAone}};
  \node[right=15em of a, inner sep=-4pt, label={below:\footnotesize\texttt{autfilt}}] (autfilt) {\rotatebox{45}{\huge\dstechnical}};

  \node at (b -| autfilt) [inner sep=-3pt, label={below:\parbox{4em}{\centering\footnotesize\texttt{autfilt\\[-4pt]-{}-cross}}}] (autfilt-cross) {\rotatebox{45}{\huge\dstechnical}};

  \node at (phi -| autfilt) [inner sep=-3pt, label={below:\footnotesize\texttt{ltl2tgba}}] (ltl2tgba) {\rotatebox{45}{\huge\dstechnical}};

  \node [above=3em of autfilt, inner sep=-3pt, label={below:\footnotesize\texttt{ic3ia}}] (ic3ia) {\rotatebox{45}{\huge\dstechnical}};

  \node at (lower |- ic3ia) [draw] (tovmt) {\rotatebox{90}{\scriptsize To VMT}} ;

  \node[draw, minimum height=12em, right=24em of center,yshift=-2em] (lift) {\rotatebox{90}{\scriptsize Lift to \HOApp}};

  \node[anchor=west] at ($(a -| lift.east)+(0.6,0)$) (aprime) {\large$\mathcal{A}'$} ;
  \node[anchor=west] at ($(b -| lift.east)+(0.6,0)$) (axb) {\large$\mathcal{A}\mathord{\otimes}\mathcal{B}$} ;
  \node[anchor=west] at ($(phi -| lift.east)+(0.6,0)$) (aphi) {\large$\mathcal{A}_{\varphi}$} ;

  \node at (ic3ia -| aprime)[anchor=west] (result) {\scriptsize\parbox{8em}{\texttt{empty}\kern6pt or\\\texttt{not empty}\kern1pt +\kern1pt cex}} ;

  \begin{scope}[->,>=latex]
  \draw (a) -- (a -| frontend.west) ;
  \draw (b) -- (b -| frontend.west) ;

  \draw (a -| frontend.east) -- (a -| lower.west) node [midway] {\scriptsize$\mathcal{A}$} ;
  \draw (b -| frontend.east) -- (b -| lower.west) node [midway] {\scriptsize$\mathcal{B}$} ;
  \draw (a -| lower.east) -- (a -| autfilt.west) node [midway] {\scriptsize$\mathcal{A}\flat$} ;

  \draw (b -| lower.east) -- (autfilt-cross.west) node [midway] {\scriptsize$\mathcal{B}\flat$} ;

  \draw ($(autfilt.east)+(0.03,0)$) -- (a -| lift.west) node [midway] {\scriptsize$\mathcal{A}'\flat$} ;

  \draw (a -| lower.east)-- ++(1.1,0) |- ($(autfilt-cross.west)+(-0,0.15)$) ;

  \draw (b -| autfilt-cross.east) -- (b -| lift.west) node [midway] {\scriptsize$(\mathcal{A}\mathord{\otimes} \mathcal{B})\flat$} ;

  \draw (a -| frontend.east)-- ++(0.8,0) |- (tovmt.west) ;

  \draw (tovmt.east |- ic3ia) -- (ic3ia.west) node[pos=0.4] {\scriptsize$\mathcal{M}, \varphi_{\mathit{Acc}}$};

  \draw (phi) -- (phi -| ltlfrontend.west) ;
  \draw (phi -| ltlfrontend.east) -- (ltl2tgba.west) node [midway] {\scriptsize$\varphi\flat$} ;

  \draw (ltl2tgba) -- (ltl2tgba -| lift.west) node [midway] (aphib) {\scriptsize$\mathcal{A}_\varphi\flat$};

  \node at (ltl2tgba -| aphib) [circle, fill=black, minimum width=4pt, inner sep=0pt] (bullet) {};

  \node at (frontend |- tovmt) (fake1) {};
  \node at (lift |- tovmt) (fake2) {};

  \draw  (autfilt -| lift.east) -- (aprime) ;
  \draw  (autfilt-cross -| lift.east) -- (axb) ;
  \draw  (ltl2tgba -| lift.east) -- (aphi) ;
  \draw  (ic3ia.east) -- (result) node [midway] {\scriptsize$\mathcal{M}\models \neg\varphi_{\mathit{Acc}}?$} ;

  \begin{scope}[on background layer]
  \draw [dashed]
  ($(args.east)+(0.02,0.05)$) -| (autfilt.north);
  \draw [dashed] ($(types.east)+(0.02,0.05)$) -| (bullet) ;
  \draw [dashed] ($(types.east)+(0.02,0.05)$) -| (ltlfrontend.south) ;
  \end{scope}

  \node[draw, rounded corners,fit={(ic3ia)(ltl2tgba)},inner sep=12pt,label={above:\scriptsize\textit{Back ends}}] (backends) {} ;

  \end{scope}
\end{tikzpicture}}
    \caption{Overview of the capabilities of the \texttt{hoapp} tool.}
    \label{fig:tool}
\end{figure}
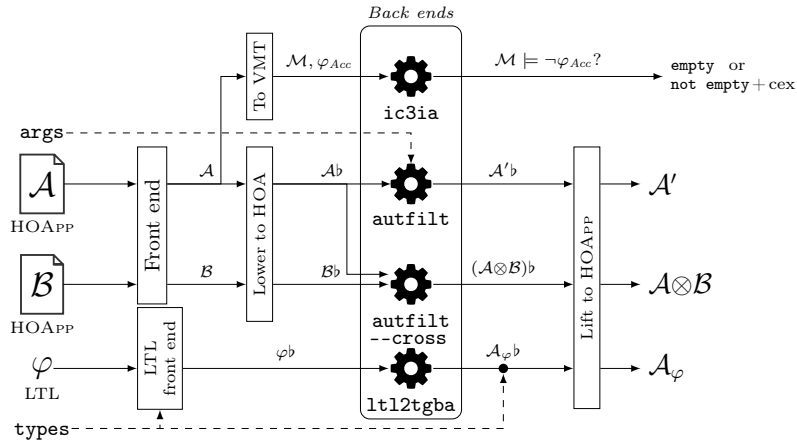

\section{The \texttt{hoapp} tool}\label{sec:tool}

We have built a prototype tool \verb|hoapp| that can parse, type-check, and process \HOApp files
through a command-line interface.
The tool is written in Python and may also be loaded as a package within other programs:
then, these can exploit its API to create and manipulate \HOApp automata programmatically.
Fig.~\ref{fig:tool} gives a graphical overview of the tool's capabilities, which we discuss below.
Support is currently limited to existential-branching automata.\footnote{The tool is open-source: \url{https://github.com/lou1306/hoapp}.}

\smallskip\noindent\emph{Lowering into \HOAone.}
The tool implements an invertible procedure to lower a \HOApp automaton $\mathcal{A}$ into a \HOAone automaton $\mathcal{A}\flat$.
To illustrate this procedure, in Listing~\ref{ex:arbiter-v1} we show the result of lowering Listing~\ref{ex:arbiter} into \HOAone.
The procedure first performs a helper step called \emph{auto-aliasing}.
This step expands all aliases in the automaton,
then defines one new alias for each variable,
and lastly applies these new aliases to the entire automaton.
This ensures consistency for the inverse procedure.
Then, we treat every predicate $p$ over non-Boolean variables in $\mathcal{A}$ as
a fresh atomic proposition in $\mathcal{A}\flat$.
For instance, in in Listing~\ref{ex:arbiter-v1} predicate $x \neq 0$ becomes a proposition named \verb|"x != 0"|.
We also replace obligation assignments by fresh APs, and add them as conjuncts to their edge's guard.
By the end of this step, every edge contains a single predicate and no obligation, and so it is compatible with \HOAone.
When emitting the HOA file,
we remove all aliases generated in the auto-aliasing step.
In turn, we attach the names, types, and controllability of $\mathcal{A}$'s variables under the items \texttt{v1pp-AP}, \texttt{v1pp-AP-type}, and \texttt{v1pp-controllable-AP}, respectively.
We also copy \texttt{assume} and \texttt{guarantee} items, as \texttt{v1pp-assume} and \texttt{v1pp-guarantee},
but we quote the formulas so that they will be treated as strings by \HOAone-compliant tools.
Note that these tools will ignore these non-standard header items, as they start with a lowercase letter.
This additional information makes the transformation invertible. (We cannot recover the original aliases of $\mathcal{A}$, but these do not affect its semantics.)

The inverse, \emph{lifting} procedure replaces the APs of $\mathcal{A}\flat$ by the original predicates;
it also separates assignments from guards, and reconstructs obligations.
We recover the original variables, their types, and any LTL formulas from the respective \texttt{v1pp-} items, and recreate aliases with the same auto-aliasing procedure described earlier.
The lowering enables reusing Spot for several operations on automata.
Spot also implements translations into a number of output formats, including GraphViz, which makes it easy to visualise \HOApp automata.

\smallskip\noindent\emph{Product of automata.}
The tool implements the product operation by lowering the two automata into \HOAone, computing their product via Spot, and lifting the result back into \HOApp.
The procedure treats variables with the same name as identical, regardless of the order in which they appear,
mimicking Spot's behaviour with \HOAone automata. We report an error if there is a mismatch in the types or controllability of variables between the two automata.

\noindent\emph{LTL translation.}
We also provide a proof-of-concept translator from LTL with predicates to B\"uchi automata, which relies on Spot's \texttt{ltl2tgba}.
The translator accepts an LTL formula whose atoms may be \HOApp comparisons over variables,
and a mapping describing the type of each variable.
The formula must use \texttt{@x} to refer to a variable $x$, and \texttt{i/r} prefixes for integer and
real literals. For instance, the formula $G(x = 0) \rightarrow F(y = 1)$ should be given as \verb|G (@x == i0) -> F (@y == i1)| (assuming \texttt{x}, \texttt{y} are declared as integers).
The translation workflow type-checks the formula against the mapping given by the user; then, it quotes all sub-expressions not natively understood by Spot and runs
the resulting formula through \texttt{ltl2tgba}; and finally, it uses the
type mapping to lift the resulting \HOAone automaton into \HOApp.

\noindent\emph{Emptiness checking.}
The tool can check an arbitrary \HOApp automaton for emptiness,
by translating the automaton into the VMT-LIB format~\cite{DBLP:conf/smt/CimattiGT22}.
In this translation, we track visits to an acceptance set $F_i$ by an acceptance bit $\texttt{ACC}_i$ which is set iff a transition belonging to $F_i$ is taken.
This lets us turn the acceptance condition $\mathit{Acc}$ into an LTL formula $\varphi_\textit{Acc}$.
For instance, $\mathrm{Fin}(i)$ and $\mathrm{Inf}(i)$ are encoded as $FG(\neg\texttt{ACC}_i)$ and $GF(\texttt{ACC}_i)$, respectively.
By this translation, we can use the state-of-the-art model checker \textsc{Ic3ia}~\cite{DBLP:conf/cav/CimattiGJRT25} to verify
whether the automaton models $\neg\varphi_\textit{Acc}$. If it does, the automaton is empty. Otherwise, the model checker returns an accepting run of the automaton as a counterexample.

\section{Conclusion}\label{sec:conclusion}

The growing interest of the research community in infinite-state systems
demands effective formalisms for automata over richer-than-Boolean alphabets.
Previous approaches have mostly focused on %
endowing automata with registers or memory cells.
In this paper, we suggested obligations as a simpler construct that increases the expressiveness
of (finite) symbolic $\omega$-automata while retaining some of their simplicity.
This simple mechanism lends itself easily to semantics for
nondeterminism, rich acceptance conditions, and universal branching.
To demonstrate the practicality of the proposal, we also introduced a
dialect of the popular HOA format to define obligation automata over integer and real variables,
as well as a prototype tool to manipulate these automata.

Questions about the decidability of emptiness and other problems for obligation
automata remain open, and their relation to register automata also requires further inquiry.
It is plain to see that we can check emptiness by reduction to model checking,
but in the future we might investigate whether any known
emptiness algorithm for $\omega$-automata may be adapted to obligation automata:
this may provide insights into the complexity of this problem.
We also plan to design a temporal logic with
obligation-like constructs, akin to TSL~\cite{DBLP:conf/cav/Finkbeiner0PS19}
but with arguably simpler semantics and a direct correspondence with (a subset of)
B\"uchi obligation automata.
On the practical side, the format that we proposed is currently limited to a small selection of operators, which we should expand in the future:
as a more general solution, we could extend the syntax to allow arbitrary SMT-LIB terms~\cite{barrett_smt-lib_2010} in the grammar of expressions,
and introduce corresponding header items to specify the theories required by the automaton.
We also intend to investigate applications of \HOApp to formal verification and synthesis.

\bibliographystyle{splncs04}
\bibliography{references}

\begin{thebibliography}{10}
\providecommand{\url}[1]{\texttt{#1}}
\providecommand{\urlprefix}{URL }
\providecommand{\doi}[1]{https://doi.org/#1}

\bibitem{DBLP:conf/cav/AzzopardiSPS25}
Azzopardi, S., Di~Stefano, L., Piterman, N., Schneider, G.: Full {LTL} {Synthesis} over {Infinite}-{State} {Arenas}. In: Piskac, R., Rakamaric, Z. (eds.) 37th {International} {Conference} on {Computer} {Aided} {Verification} ({CAV}). {LNCS}, vol. 15934, pp. 274--297. Springer (2025). \doi{10.1007/978-3-031-98685-7_13}

\bibitem{DBLP:conf/cav/BabiakBDKKM0S15}
Babiak, T., Blahoudek, F., Duret-Lutz, A., Klein, J., Kretínský, J., Müller, D., Parker, D., Strejcek, J.: The {Hanoi} {Omega}-{Automata} {Format}. In: Kroening, D., Pasareanu, C.S. (eds.) 27th {International} {Conference} on {Computer} {Aided} {Verification} ({CAV}). {LNCS}, vol.~9206, pp. 479--486. Springer (2015). \doi{10.1007/978-3-319-21690-4_31}

\bibitem{barrett_smt-lib_2010}
Barrett, C.W., Stump, A., Tinelli, C.: The {SMT}-{LIB} {Standard} - {Version} 2.0 (2010), \url{http://smtlib.cs.uiowa.edu/papers/smt-lib-reference-v2.0-r10.12.21.pdf}

\bibitem{biere_aiger_2007}
Biere, A.: The {AIGER} {And}-{Inverter} {Graph} ({AIG}) {Format} {Version} 20071012. Technical {Report}, Johannes Kepler University (2007). \doi{10.35011/fmvtr.2007-1}

\bibitem{biere_aiger_2011}
Biere, A., Heljanko, K., Wieringa, S.: {{AIGER}} 1.9 and beyond. Technical {{Report}}~11/2, Johannes Kepler University, Linz, Austria (2011), \url{https://doi.org/10.35011/fmvtr.2011-2}

\bibitem{DBLP:conf/icalp/BokerKR10}
Boker, U., Kupferman, O., Rosenberg, A.: Alternation removal in {{B\"uchi}} automata. In: Abramsky, S., Gavoille, C., Kirchner, C., auf {der Heide}, F.M., Spirakis, P.G. (eds.) 37th {{International Colloquium}} on {{Automata}}, {{Languages}} and {{Programming}} ({{ICALP}}). {{LNCS}}, vol.~6199, pp. 76--87. Springer, Bordeaux, France (2010). \doi{10.1007/978-3-642-14162-1_7}

\bibitem{DBLP:conf/cav/CavadaCDGMMMRT14}
Cavada, R., Cimatti, A., Dorigatti, M., Griggio, A., Mariotti, A., Micheli, A., Mover, S., Roveri, M., Tonetta, S.: The {nuXmv} {Symbolic} {Model} {Checker}. In: Biere, A., Bloem, R. (eds.) 26th {International} {Conference} on {Computer} {Aided} {Verification} ({CAV}). {LNCS}, vol.~8559, pp. 334--342. Springer (2014). \doi{10.1007/978-3-319-08867-9_22}

\bibitem{DBLP:conf/cav/CimattiGJRT25}
Cimatti, A., Griggio, A., Johannsen, C., Rozier, K.Y., Tonetta, S.: Infinite-{{State Liveness Checking}} with rlive. In: Piskac, R., Rakamaric, Z. (eds.) 37th {{International Conference}} on {{Computer Aided Verification}} ({{CAV}}). {{LNCS}}, vol. 15931, pp. 215--236. Springer, Zagreb, Croatia (2025). \doi{10.1007/978-3-031-98668-0_11}

\bibitem{DBLP:conf/smt/CimattiGT22}
Cimatti, A., Griggio, A., Tonetta, S.: The {{VMT-LIB Language}} and {{Tools}}. In: D{\'e}harbe, D., Hyv{\"a}rinen, A.E.J. (eds.) 20th Workshop on {{Satisfiability Modulo Theories}} ({{SMT}}). {{CEUR Workshop Proceedings}}, vol.~3185, pp. 80--89. CEUR-WS.org (2022), \url{https://ceur-ws.org/Vol-3185/extended9547.pdf}

\bibitem{DBLP:conf/cav/ClarkeMCH96}
Clarke, E.M., McMillan, K.L., Campos, S.V.A., {Hartonas-Garmhausen}, V.: Symbolic {{Model Checking}}. In: Alur, R., Henzinger, T.A. (eds.) 8th {{International Conference}} on {{Computer Aided Verification}} ({{CAV}}). {{LNCS}}, vol.~1102, pp. 419--427. Springer (1996). \doi{10.1007/3-540-61474-5_93}

\bibitem{DBLP:conf/lics/DemriL06}
Demri, S., Lazic, R.: {{LTL}} with the {{Freeze Quantifier}} and {{Register Automata}}. In: 21th {{Symposium}} on {{Logic}} in {{Computer Science}} ({{LICS}}). pp. 17--26. IEEE, Seattle, WA, USA (2006). \doi{10.1109/LICS.2006.31}

\bibitem{DBLP:journals/scp/EmersonL87}
Emerson, E.A., Lei, C.L.: Modalities for model checking: Branching time logic strikes back. Science of Computer Programming  \textbf{8}(3),  275--306 (1987). \doi{10.1016/0167-6423(87)90036-0}

\bibitem{DBLP:conf/fossacs/FinkbeinerHP22}
Finkbeiner, B., Heim, P., Passing, N.: Temporal stream logic modulo theories. In: Bouyer, P., Schr{\"o}der, L. (eds.) 25th {{International Conference}} on {{Foundations}} of {{Software Science}} and {{Computation Structures}} ({{FOSSACS}}). {{LNCS}}, vol. 13242, pp. 325--346. Springer, Munich, Germany (2022). \doi{10.1007/978-3-030-99253-8_17}

\bibitem{DBLP:conf/cav/Finkbeiner0PS19}
Finkbeiner, B., Klein, F., Piskac, R., Santolucito, M.: Temporal {Stream} {Logic}: {Synthesis} {Beyond} the {Bools}. In: Dillig, I., Tasiran, S. (eds.) 31st {International} {Conference} on {Computer} {Aided} {Verification} ({CAV}). {LNCS}, vol. 11561, pp. 609--629. Springer (2019). \doi{10.1007/978-3-030-25540-4_35}

\bibitem{DBLP:journals/pacmpl/HeimD24}
Heim, P., Dimitrova, R.: Solving {Infinite}-{State} {Games} via {Acceleration}. Proceedings of the ACM on Programming Languages  \textbf{8}(POPL),  1696--1726 (2024). \doi{10.1145/3632899}

\bibitem{DBLP:conf/cav/HeimD25}
Heim, P., Dimitrova, R.: Issy: {{A}} comprehensive tool for specification and synthesis of infinite-state reactive systems. In: Piskac, R., Rakamaric, Z. (eds.) 37th International Conference on {{Computer Aided Verification}} ({{CAV}}). {{LNCS}} in {{Computer Science}}, vol. 15934, pp. 298--312. Springer (2025). \doi{10.1007/978-3-031-98685-7_14}

\bibitem{DBLP:conf/cav/JohannsenNDISTVR24}
Johannsen, C., Nukala, K., Dureja, R., Irfan, A., Shankar, N., Tinelli, C., Vardi, M.Y., Rozier, K.Y.: The {MoXI} {Model} {Exchange} {Tool} {Suite}. In: Gurfinkel, A., Ganesh, V. (eds.) 36th {International} {Conference} on {Computer} {Aided} {Verification} ({CAV}). {LNCS}, vol. 14681, pp. 203--218. Springer, Montreal, QC, Canada (2024). \doi{10.1007/978-3-031-65627-9_10}

\bibitem{DBLP:journals/tcs/KaminskiF94}
Kaminski, M., Francez, N.: Finite-memory automata. Theoretical Computer Science  \textbf{134}(2),  329--363 (1994). \doi{10.1016/0304-3975(94)90242-9}

\bibitem{DBLP:conf/fmcad/MaderbacherB22}
Maderbacher, B., Bloem, R.: Reactive {Synthesis} {Modulo} {Theories} using {Abstraction} {Refinement}. In: Griggio, A., Rungta, N. (eds.) 22nd {Conference} on {Formal} {Methods} in {Computer}-{Aided} {Design} ({FMCAD}). pp. 315--324. IEEE (2022). \doi{10.34727/2022/isbn.978-3-85448-053-2_38}

\bibitem{DBLP:conf/cav/NiemetzPWB18}
Niemetz, A., Preiner, M., Wolf, C., Biere, A.: Btor2 , {{BtorMC}} and {{Boolector}} 3.0. In: Chockler, H., Weissenbacher, G. (eds.) 30th {{International Conference}} on {{Computer Aided Verification}} ({{CAV}}). {{LNCS}}, vol. 10981, pp. 587--595. Springer (2018). \doi{10.1007/978-3-319-96145-3_32}

\bibitem{DBLP:conf/focs/Pnueli77}
Pnueli, A.: The {{Temporal Logic}} of {{Programs}}. In: 18th {{Annual Symposium}} on {{Foundations}} of {{Computer Science}} ({{FOCS}}). pp. 46--57. IEEE (1977). \doi{10.1109/SFCS.1977.32}

\bibitem{DBLP:journals/corr/abs-1912-05793}
Pérez, G.A.: The {Extended} {HOA} {Format} for {Synthesis}. CoRR  \textbf{abs/1912.05793} (2019), \url{http://arxiv.org/abs/1912.05793}

\bibitem{DBLP:conf/lics/VardiW86}
Vardi, M.Y., Wolper, P.: An automata-theoretic approach to automatic program verification (preliminary report). In: Symposium on Logic in Computer Science ({{LICS}}). pp. 332--344. IEEE, Cambridge, MA, USA (1986)

\end{thebibliography}

\clearpage\appendix

\section{Negative obligations}\label{apx:nobl}
\newcommand*{\no}{\overline{o}}
In this appendix, we try and further extend the obligation automata
framework with an extra construct that may be useful for complementation.
We chose to relegate this extension to an appendix due to space reasons,
to make our main treatment of obligation automata easier to follow, and
(on a more practical side) because support for this feature has not yet
been implemented in \HOApp and its companion tool.

We now assume that every edge may
be additionally labelled by a \emph{negative} obligation $\no$.
This is a set of pairs from $V \times \mathcal{T}(V)$,
which places \emph{inequality} constraints on the
next input symbol, as opposed to the equality constraints set by
obligations.
We write a negative obligation as
$\no = \{x_1 \nmapsto t_1, x_2 \nmapsto t_2, \ldots\}$, where the
same variable $x$ may appear any number of times,
and define $\SEM{\no}{v}$ as $\bigwedge_i x_i \neq \SEM{t_i}v$,
or $\textit{true}$ if $\no$ is the empty set.

Then, a sequence of edges $e_1e_2\ldots$
with negative obligations $\no_1\no_2\ldots$
is a run for a word $\mathbf{v}$ iff, in addition to
the definition in Section~\ref{sec:oblaut},
it also satisfies the additional constraint
$v_{i+1} \models \SEM{\no_i}v_i$, for every $i$.
When dealing with products, we simply take the
union of negative obligations: thus,
$(q_1, \gamma_1, o_1, \no_1, q'_1)
\times
(q_2, \gamma_2, o_2, \no_2, q'_2)
= (q_1, q_2) \xrightarrow{\gamma_1\land\gamma_2, o_1 \times o_2, \no_1 \cup \no_2} (q'_1, q'_2)
$. Universal branching is adapted similarly.
We conjecture that Theorems~\ref{thm:intersect} and~\ref{thm:fragment} readily extend to negative obligations.

With negative obligations, we can
recognise the complement of~$\mathcal{L}$
(the language of sequences where $x$ continuously grows by one unit,
recognised by the automaton in Fig.~\ref{fig:sequence2})
without relying on extra variables, for instance via the B\"uchi
automaton in Fig.~\ref{fig:nobl}.
When reading words from~$\mathcal{L}$, this automaton simply loops
in the initial, non-accepting state.
To reach the accepting state, we must at some point read a
valuation $v_i$ such that $x$ is \emph{not} equal to $v_{i-1}(x)+1$;
there is a run where reading $v_{i-1}$ takes the edge
with the negative obligation, and $v_i$ does satisfy
the constrain $v_i(x) \neq v_{i-1}(x) + 1$,
so it is successfully read by the automaton.
From that point on, the automaton may forever stay in 
the accepting sink, and thus recognise that the word is in
$\mathcal{L}^c$. Words from $\mathcal{L}$, instead, keep the
automaton in the non-accepting initial state.

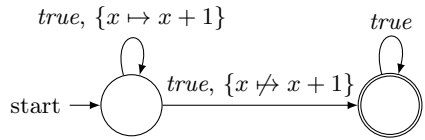
\begin{figure}[b]
\centering
\begin{tikzpicture}[auto, every initial by arrow/.style={>=latex}]
    \node [state, initial] (q0){};
    \node [state, accepting, right=8em of q0] (q1){};
    \begin{scope}[->, >=latex]
        \draw (q0) to[loop above] node[midway] {$\textit{true}$, $\{x \mapsto x+1\}$} (q0);
        \draw (q0) to node[midway] {$\textit{true}$, $\{x \nmapsto x+1\}$} (q1);
        \draw (q1) to[loop above] node[midway] {$\textit{true}$} (q1);
    \end{scope}
\end{tikzpicture}
\caption{B\"uchi (negative-) obligation automaton for $\mathcal{L}^c$.}\label{fig:nobl}
\end{figure}

\begin{figure}[tb]
\centering
\subfloat[A word where $x$ always changes is accepted.]{%
\begin{tikzpicture}[auto, every initial by arrow/.style={>=latex}]

    \node (n1) {};
    \node[below=4em of n1] (n2) {};
    \node[below=4em of n2] (n3) {};
    \node[right=8em of n1] (v1) {\texttt{1}};
    \node[above right= -0.5em and 8em of n2] (v2) {\texttt{2}};
    \node[above right= -0.5em and 8em of n3, label=below:{$\vdots$}] (v3) {\texttt{3}};
    \node[above=1em of v1] (x) {$x$};

    \node [state, initial, above= .5em of n1,initial text=] (q0){$q_0$};
    \node [state, accepting, below=.5em of n1] (q11) {$q_1$};

    \node [state, accepting, below left=.5em and 2em of n2] (q12x) {$q_{1,2}$};
    \node [state, accepting, below right=.5em and 2em of n2] (q12y) {$q_{1,2}$};

    \node [label=below:{$\vdots$},state, accepting, below left=.5em and 4em of n3] (q12xx) {$q_{1,2}$};
    \node [label=below:{$\vdots$},state, accepting, below left=.5em and .5em of n3] (q12xy) {$q_{1,2}$};
    \node [label=below:{$\vdots$},state, accepting, below right=.5em and .5em of n3] (q12yx) {$q_{1,2}$};
    \node [label=below:{$\vdots$},state, accepting, below right=.5em and 4em of n3] (q12yy) {$q_{1,2}$};

    \begin{scope}[->, >=latex]
    \draw (q0) -- (q11) node[midway,left] {$y\mapsto x$};
    \draw (q11) -- (q12x) node[midway,left] {$y \mapsto x$};
    \draw (q11) -- (q12y) node[midway,right] {$y \mapsto y$};

    \draw (q12x) -- (q12xx) node[midway,left] {\scriptsize$y \mapsto x$};
    \draw (q12x) -- (q12xy) node[midway,right,xshift=-2pt] {\scriptsize$y \mapsto y$};
    \draw (q12y) -- (q12yx) node[midway,left,xshift=2pt] {\scriptsize$y \mapsto x$};
    \draw (q12y) -- (q12yy) node[midway,right] {\scriptsize$y \mapsto y$};
    \end{scope}
\end{tikzpicture}}
\hfill
\subfloat[A word where $x$ takes the same value twice is rejected: the right-most branch can never leave $q_3$.]{%
\begin{tikzpicture}[auto, every initial by arrow/.style={>=latex}]

    \node (n1) {};
    \node[below=4em of n1] (n2) {};
    \node[below=4em of n2] (n3) {};
    \node[right=8em of n1] (v1) {\texttt{1}};
    \node[above right= -0.5em and 8em of n2] (v2) {\texttt{2}};
    \node[above right= -0.5em and 8em of n3, label=below:{$\vdots$}] (v3) {\texttt{1}};
    \node[above=1em of v1] (x) {$x$};

    \node [state, initial, above= .5em of n1,initial text=] (q0){$q_0$};
    \node [state, accepting, below=.5em of n1] (q11) {$q_1$};

    \node [state, accepting, below left=.5em and 2em of n2] (q12x) {$q_{1,2}$};
    \node [state, accepting, below right=.5em and 2em of n2] (q12y) {$q_{1,2}$};

    \node [label=below:{$\vdots$},state, accepting, below left=.5em and 4em of n3] (q12xx) {$q_{1,2}$};
    \node [label=below:{$\vdots$},state, accepting, below left=.5em and .5em of n3] (q12xy) {$q_{1,2}$};
    \node [label=below:{$\vdots$},state, below right=.5em and 2em of n3] (q12yx) {$q_{3}$};

    \begin{scope}[->, >=latex]
    \draw (q0) -- (q11) node[midway,left] {$y\mapsto x$};
    \draw (q11) -- (q12x) node[midway,left] {$y \mapsto x$};
    \draw (q11) -- (q12y) node[midway,right] {$y \mapsto y$};

    \draw (q12x) -- (q12xx) node[midway,left] {\scriptsize$y \mapsto x$};
    \draw (q12x) -- (q12xy) node[midway,right,xshift=-2pt] {\scriptsize$y \mapsto y$};
    \draw (q12y) -- (q12yx) node[midway,right,xshift=2pt] {\scriptsize$y=x$};

    \end{scope}
\end{tikzpicture}}
\caption{Sample runs in the automaton from Figure~\ref{fig:nonce}.
Values of $x$ in the word corresponding to each run are shown to the right. Although runs are defined as DAGs, in this figure we
visualise them as trees by duplicating some nodes. We also omit the guard $y\neq x$ from
most edges, to reduce clutter.}\label{fig:nonceruns}
\end{figure}
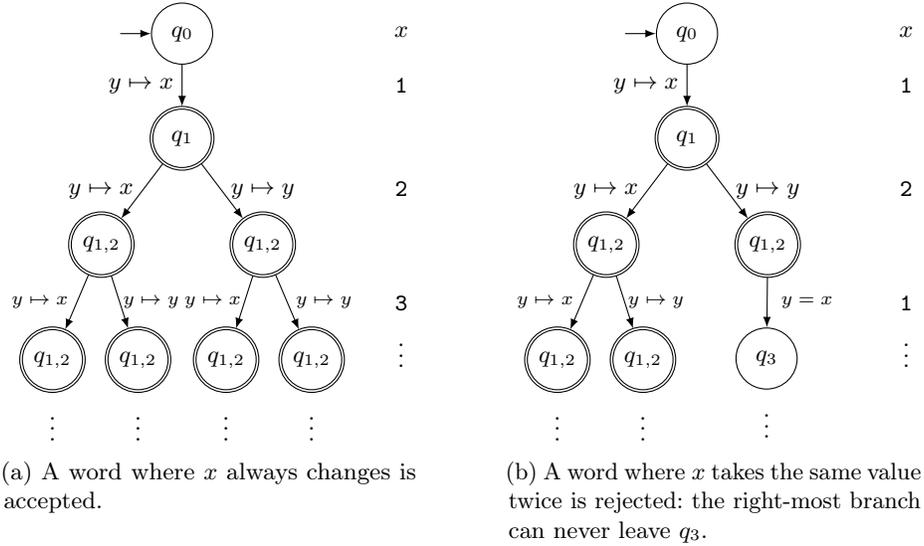

\end{document}